\documentclass[aps,pra,notitlepage,twocolumn,10pt,floatfix]{revtex4-2}
\pdfoutput=1
\usepackage[utf8]{inputenc}
\usepackage[english]{babel}
\usepackage[T1]{fontenc}
\usepackage[no-theorems]{mplmacros}
\usepackage{thmtools, thm-restate}
\newtheorem*{definition*}{Definition}
\newtheorem*{proposition*}{Proposition}
\newtheorem*{theorem*}{Theorem}
\newtheorem*{lemma*}{Lemma}
\renewcommand{\paragraph}[1]{\addcontentsline{toc}{section}{#1}\emph{#1.}---}

\begin{document}
\title{Contextuality as a Precondition for Quantum Entanglement}

\author{Martin Pl\'{a}vala}
\affiliation{Naturwissenschaftlich-Technische Fakult\"{a}t, Universit\"{a}t Siegen, 57068 Siegen, Germany}

\author{Otfried G\"{u}hne}
\affiliation{Naturwissenschaftlich-Technische Fakult\"{a}t, Universit\"{a}t Siegen, 57068 Siegen, Germany}

\begin{abstract}
Quantum theory features several phenomena which can be considered as resources for information processing tasks. Some of these effects, such as entanglement, arise in a nonlocal scenario, where a quantum state is distributed between different parties. Other phenomena, such as contextuality, can be observed if quantum states are prepared and then subjected to sequences of measurements. We use robust remote state preparation to connect the nonlocal and sequential scenarios and provide an intimate connection between different resources: We prove that entanglement in a nonlocal scenario can arise only if there is preparation and measurement contextuality in the corresponding sequential scenario and that the absence of entanglement implies the absence of contextuality. As a direct consequence, our result allows us to translate any inequality for testing preparation and measurement contextuality into an entanglement test; in addition, entanglement witnesses can be used to design novel contextuality inequalities.
\end{abstract}

\maketitle
\paragraph{Introduction}
Quantum information science bears the promise to lead to novel ways of information processing which are superior to classical methods. This begs the question which quantum phenomena are responsible for the quantum advantage and which resources are needed to overcome classical limits. There are two main scenarios where genuine quantum effects are studied. First, in the nonlocal scenario (NLS), two parties, Alice and Bob, share a bipartite quantum state $\rho_{AB}$ and perform different measurements on it. This leads to a joint probability distribution for the possible outcomes. Second, in the sequential scenario (SQS), Bob prepares some quantum state $\sigma$ and transmits the quantum state to Alice, who performs a measurement. Clearly, these scenarios are connected: In the NLS Alice and Bob can, using classical communication, postselect on the outcome of Bob's measurement, so that Bob remotely prepares the state $\sigma$ for Alice, see also Fig.~\ref{fig:results}.

In both scenarios, several notions of non-classicality are known and have been identified as resources for special tasks. For the NLS a major example is entanglement which arises if the quantum state cannot be generated by local operations and classical communication \cite{GuhneToth-entanglement, HorodeckiHorodeckiHorodeckiHorodecki-entanglement}. Entanglement has been identified as a resource for tasks like quantum key distribution \cite{CurtyLewensteinLu-entanglementForQKD} or quantum metrology \cite{PezzeSmerzi-entanglement, TothApellaniz-metrology}. Other examples of non-classicality in the NLS are quantum steering \cite{UolaCostaNguyenGuhne-steering} and Bell nonlocality \cite{BrunnerCavalcantiPironioScaraniWehner-BellNonlocality}.

For the SQS a major notion of non-classicality is contextuality \cite{Spekkens-contextuality,BudroniCabelloGuhneKleinmann-contextuality}; we will consider the preparation and measurement (P\&M) contextuality, sometimes also called simplex embeddability \cite{SchmidSelbyWolfeKunjwalSpekkens-noncontextuality, SchmidSelbyPuseySpekkens-nonconModels, SelbySchmidWolfeSainzKunjwalSpekkens-contextuality, SelbySchmidWolfeSainzKunjwalSpekkens-GPTfragments}. Contextuality can be viewed as a resource in various tasks, such as quantum state discrimination \cite{SchmidSpekkens-stateDiscrimination}, cryptography \cite{SpekkensBuzacottKeehnTonerPryde-parityObliviousMultiplexing,AmbainisBanikChaturvediKravchenkoRai-RACcontextuality,ChaillouxKerenidisKunduSikora-RACcontextuality}, quantum computation \cite{HobanCampbellLoukopoulosBrowne-MBQCcontextuality,HowardWallmanVeitchEmerson-ComputationContextuality,Raussendorf-MBQCcontextuality,SchmidHaoxingSelbyPusey-noncontextual,Shahandeh-contextualityComputing}, and metrology \cite{Lostaglio-contextuality}.

\begin{figure}
\includegraphics[width=\linewidth]{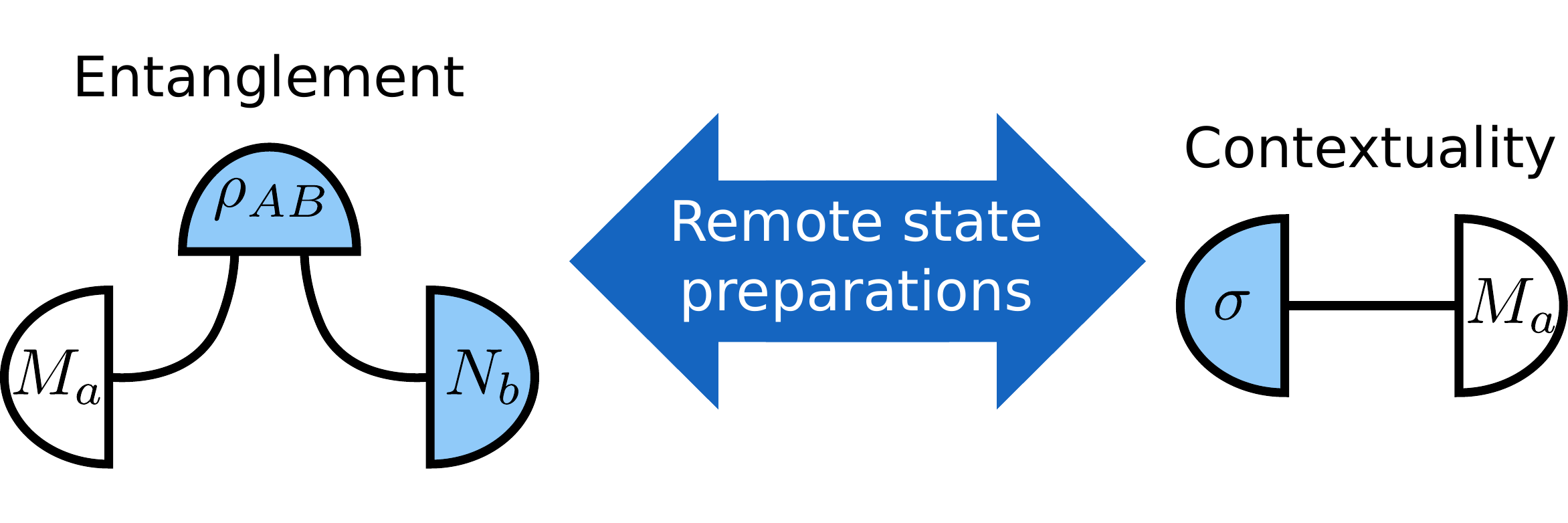}
\caption{Connections between entanglement in the nonlocal scenario and contextuality in the sequential scenario. In the nonlocal scenario, one considers probability distributions of the type $p(a,b|\rho_{AB})=\Tr[(M_a \otimes N_b) \rho_{AB}]$, where $M_a$ and $N_b$ describe measurements. In the sequential scenario, Alice prepares a quantum state $\sigma$, transmits it to Bob, leading to the probabilities $p(a|\sigma)=\Tr(\sigma M_a)$. We show that entanglement in the nonlocal scenario is practically equivalent to noncontextuality in the sequential scenario, if the state $\sigma \sim \Tr_B[(\I_a \otimes N_b) \rho_{AB}]$ is prepared remotely by performing a measurement on Bob's part of the nonlocal scenario.}
\label{fig:results}
\end{figure}

Are there any connections between quantum resources arising in the nonlocal and the sequential scenario? This is a key question for understanding the quantum advantage in information processing. For quantum key distribution it was already observed some time ago that prepare and measure schemes (like the BB84 protocol) can be mapped to entanglement-based schemes, which allows for a common security analysis of both scenarios based on entanglement theory \cite{CurtyLewensteinLu-entanglementForQKD, BennettBrassardMermin-QKD}. More recently, the notion of remote state preparation was used to show that steerability of a quantum state $\rho_{AB}$ corresponds to preparation noncontextuality \cite{Plavala-contextualitySteering}, while steerability of an assemblage corresponds to measurement noncontextuality \cite{TavakoliUola-contextuality}, see also Ref.~\cite{SelbySchmidWolfeSainzKunjwalSpekkens-contextuality} for a discussion.

In this paper we show that entanglement in the NLS corresponds to P\&M contextuality in the SQS. We use the fact that any bipartite quantum state gives rise to some sequential scenario by using a kind of remote state preparation \cite{BennettDiVincenzoShorSmolinTerhalWootters-remoteStatePreparation, JevticPuseyJenningsTerry-steeringEllipsoids}. Then, P\&M contextuality in this SQS is a precondition of entanglement of the bipartite state. Our results show that P\&M contextuality is a notable model of classicality in the SQS since it corresponds to entanglement in NLS and they imply that one can map noncontextuality inequalities to entanglement witnesses and that one can use classes of entanglement witnesses to design noncontextuality inequalities. Additionally, we also discuss the impact of our results on quantum key distribution. Our research was motivated by recent findings on a connection between noncontextuality and the mathematical notion of unit separability \cite{GittonWoods-unitSeparability}. Our results significantly differ from the previously known connections between Kochen-Specker contextuality and Bell nonlocality \cite{AbramskyBrandenburger-contextuality,AcinFritzLeverrierSainz-contextuality,CabelloSeveriniWinter-contextuality,AmaralCuhna-contextuality,WagnerBarbosaGalvao-graphIneq} since not every entangled state exhibits Bell nonlocality and despite the similar name, Kochen-Specker contextuality and Spekkens contextuality are strictly different operational concepts exhibiting different non-classical properties, even though there are known connections between Spekkens contextuality and Bell nonlocality \cite{LiangSpekkensWiseman-BellContextuality,SchmidSpekkensWolfe-noncontextualityIneq,SchmidSpekkens-stateDiscrimination}.

\paragraph{Entanglement and remote preparations}
Assume that two remote parties, Alice and Bob, share a bipartite quantum state $\rho_{AB}$. We then say that the bipartite state $\rho_{AB}$ is entangled if it cannot be prepared using local operations and classical communication \cite{HeinosaariZiman-MLQT}. This is the same as requiring that the state $\rho_{AB}$ is not separable, i.e., there is no decomposition of the form $\rho_{AB} = \sum_{i} p_i \sigma_i^A \otimes \sigma_i^B$, where the $p_i$ form a probability distribution and $\sigma_i^A$ and $\sigma_i^B$ are some states of Alice's and Bob's system.

Given a bipartite state $\rho_{AB}$, Bob can apply a measurement to his part of the system and announce the outcome, which results in remotely preparing the state $\Tr_B[(\I_A \otimes E_B) \rho_{AB}]$ for Alice. While a single remotely prepared state does not capture the properties of the bipartite state $\rho_{AB}$ shared between Alice and Bob, it is intuitive that the set of all possible remotely preparable states should have some properties based on whether the shared bipartite state $\rho_{AB}$ is entangled or not. In order to investigate this, we denote by $\Lambda_A(\rho_{AB})$ the set of all possible states that Bob can remotely prepare for Alice using the shared bipartite state $\rho_{AB}$. Mathematically $\Lambda_A(\rho_{AB})$ is defined as
\begin{equation}
\begin{split}
\Lambda_A(\rho_{AB}) = \{ &\sigma_A \in \dens(\Ha_A): \exists E_B \geq 0 \; \text{s.t.} \\
&\sigma_A = \Tr_B[(\I_A \otimes E_B) \rho_{AB}] \},
\end{split}
\end{equation}
where $\Ha_A$ is the Hilbert space corresponding to Alice's system, $\dens(\Ha_A)$ is the set of density matrices on $\Ha_A$, and $E_B \geq 0$ means that $E_B$ is a positive semidefinite operator. $\Lambda_B(\rho_{AB})$ is defined analogously. Note that these sets of quantum states play a role in recent approaches to tackle the problem of quantum steering \cite{JevticPuseyJenningsTerry-steeringEllipsoids, NguyenNguyenGuhne-steering}.

\paragraph{Contextuality}
There are several notions of contextuality, here we are interested in the so-called preparation \& measurement (P\&M) noncontextual models in the sense of Spekkens \cite{Spekkens-contextuality,BudroniCabelloGuhneKleinmann-contextuality}. This notion of contextuality is operationally well defined and straightforward to generalize to operational and probabilistic theories \cite{SchmidSelbySpekkens-contextualityDefense}. The details of this approach were previously heavily discussed in the literature, we will provide only an abridged introduction and refer the reader to existing literature for in-depth treatment \cite{Spekkens-contextuality, BudroniCabelloGuhneKleinmann-contextuality, SchmidSelbyWolfeKunjwalSpekkens-noncontextuality, SchmidSelbyPuseySpekkens-nonconModels, SelbySchmidWolfeSainzKunjwalSpekkens-contextuality, SelbySchmidWolfeSainzKunjwalSpekkens-GPTfragments, SchmidSelbySpekkens-contextualityDefense, Plavala-contextualitySteering, MullerGarner-testingQT, AcinFritzLeverrierSainz-contextuality, LiangSpekkensWiseman-BellContextuality, SchmidSpekkensWolfe-noncontextualityIneq, SchmidSelbyPuseySpekkens-nonconModels, TavakoliUola-contextuality, Spekkens-ontological}.

In this approach one considers equivalence classes of preparations, in quantum theory represented by density matrices, and equivalence classes of measurements, in quantum theory represented by positive operator-valued measures (POVMs). A hidden variable model for a given set of preparations and measurements is given by the probability $p(\lambda|\rho)$ of preparing the the hidden variable $\lambda$ given the density matrix $\rho$ and by the response functions, that is, by probability $p(a|\oM, \lambda)$ of observing the outcome $a$ given the hidden variable $\lambda$ and the POVM $\oM$ was measured. The probability $p(a|\oM, \rho)$ of observing the outcome $a$ if $\rho$ was prepared and $M$ was measured must satisfy $p(a|\oM, \rho) = \sum_\lambda p(\lambda|P) p(a|\oM, \lambda)$.

In order to obtain the definition of P\&M contextuality, we will assume that the distribution of the hidden variable $p(\lambda|\rho)$ can be obtained only by having access to a single copy of the prepared system and that the response function can be obtained only by having access to single use of the measurement device; operational arguments imply that $p(\lambda|\rho)$ must be a linear function of the density matrix and the response function must be a linear function of the POVM \cite{Plavala-contextualitySteering, MullerGarner-testingQT}.

Given that we can prepare only states $\rho \in K \subset \dens(\Ha)$, where $K$ is a convex subset of density matrices, and that we can measure all measurements, a quantum system has a P\&M noncontextual model if for every density matrix $\rho \in K$ and for every POVM $\oM = \{M_a\}$, $0 \leq M_a \leq \I$, $\sum_a M_a = \I$, we have $\Tr(\rho M_a) = \sum_\lambda p(\lambda|\rho) p(a|\lambda, \oM)$ where $p(\lambda|\rho)$ is a linear function of $\rho$ and $p(a|\lambda, \oM)$ is a linear function of $M_a$. It follows that there are operators $N_\lambda$ such that $p(\lambda | \rho) = \Tr(\rho N_\lambda)$ that satisfy the positivity and normalization conditions:
\begin{align} \label{eq:contextuality-positivityNormalization}
 & \Tr(\rho N_\lambda) \geq 0
 & & \mbox{ and }
 & & \sum_\lambda \Tr(\rho N_\lambda) = 1,
\end{align}
for all $\rho \in K$. Note that this does not imply that the $N_\lambda$ form a POVM. In fact, whenever $K$ is a strict subset of the density matrices, then according to the hyperplane separation theorem \cite{Rockafellar-convex} there is some $N_\lambda$ that satisfies the positivity condition in Eq.~\eqref{eq:contextuality-positivityNormalization} that is not a positive semidefinite operator. Analogically, since $p(a|\lambda, \oM)$ is a linear function of $M_a$, we have $p(a|\lambda, \oM) = \Tr(\omega_\lambda M_a)$, where $\omega_\lambda \in \dens(\Ha)$. Putting everything together, we get the following definition that is sufficient for our needs:
\begin{definition*}[P\&M noncontextual model]
There exists preparation \& measurement (P\&M) noncontextual model for $K \subset \dens(\Ha)$ if for every $\rho \in K$ and every POVM $\oM = \{M_a\}$ we have
\begin{equation}
\label{eq:contextuality-def}
\Tr(\rho M_a) = \sum_\lambda \Tr(\rho N_\lambda) \Tr(\omega_\lambda M_a),
\end{equation}
where $\omega_\lambda \in \dens(\Ha)$ and $N_\lambda$ are operators satisfying \eqref{eq:contextuality-positivityNormalization}.
\end{definition*}
As in our definition, we always consider the P\&M noncontextual model of $K \subset \dens(\Ha)$ with respect to all measurements since our results connect entanglement and P\&M noncontextuality with respect to all measurements.

\paragraph{Main Results}
We can directly formulate our first main result:

\begin{restatable}{theorem}{resultsNCtoSEP} \label{thm:results-NCtoSEP}
Let $\rho_{AB}$ be a bipartite quantum state and assume that there exists a P\&M noncontextual model for the set of states $\Lambda_A(\rho_{AB})$ and all possible measurements. Then, $\rho_{AB}$ is separable.
\end{restatable}
The proof is given in the Appendix~\ref{appendix:thm-results-NCtoSEP}. The underlying idea is that in the P\&M noncontextual model from Eq.~(\ref{eq:contextuality-def}) the term $\Tr(\omega_\lambda M_a)$ can be interpreted as measurement on Alice's side of a separable decomposition $\rho_{AB} = \sum_\lambda \omega_\lambda \otimes K_\lambda$, where Bob's parts $K_\lambda$ in the decomposition are given by $K_\lambda= \Tr_A[(N_\lambda \otimes \I_B) \rho_{AB}]$.

In the following theorem we will prove that a separable state $\rho_{AB} = \sum_\lambda \omega_\lambda \otimes K_\lambda$ yields remotely preparable set $\Lambda_A(\rho_{AB})$ with a P\&M noncontextual model if the operators $K_\lambda$ in the separable decomposition can be chosen such that they belong to the linear hull of $\Lambda_B(\rho_{AB})$. 
We will get rid of this condition later by considering robust remote preparations.

\begin{restatable}{theorem}{resultsSEPtoNC} \label{thm:results-SEPtoNC}
Let $\rho_{AB}$ be a separable bipartite quantum state with the decomposition $\rho_{AB} = \sum_\lambda \omega_\lambda \otimes K_\lambda$, where $\omega_\lambda \geq 0$, $\Tr(\omega_\lambda) = 1$ and $K_\lambda \geq 0$. Assume that $K_\lambda$ belongs to the linear hull of $\Lambda_B(\rho_{AB})$ for all $\lambda$. Then there exists P\&M noncontextual model for $\Lambda_A(\rho_{AB})$.
\end{restatable}

The proof can be found in the Appendix~\ref{appendix:thm-results-SEPtoNC}. In the dimension restricted case where we assume that both Alice and Bob have locally only a single qubit, $\dim(\Ha_A) = \dim(\Ha_B) = 2$, one can see from the separability criteria presented in \cite{JevticPuseyJenningsTerry-steeringEllipsoids} that the condition on the $K_\lambda$ to be in the linear hull of $\Lambda_B(\rho_{AB})$ is not necessary and the state is separable if and only if there exists P\&M noncontextual model for $\Lambda_A(\rho_{AB})$. The following example shows that Theorem~\ref{thm:results-SEPtoNC} cannot be directly extended to all states. Consider the separable qubit-ququart state
\begin{equation} \label{eq:results-exm-SEPgivesC}
\begin{split}
\rho_{AB} &= \frac{1}{4}( \ketbra{0} \otimes \ketbra{00} + \ketbra{1} \otimes \ketbra{01} \\
&+ \ketbra{+} \otimes \ketbra{10} + \ketbra{-} \otimes \ketbra{11} ).
\end{split}
\end{equation}
This does not meet the condition in Theorem \ref{thm:results-SEPtoNC}; at least for the decomposition in Eq.~(\ref{eq:results-exm-SEPgivesC}) this is obvious: $\Lambda_B(\rho_{AB})$ is three-dimensional, while there are four linearly independent $K_\lambda$. On the other hand, $\Lambda_A(\rho_{AB}) = \conv(\{\ketbra{0}, \ketbra{1}, \ketbra{+}, \ketbra{-} \})$ does not have a P\&M noncontextual model, see Eq.~ \eqref{eq:ineqToWitness-example} below.

We now want to get a version of Theorem~\ref{thm:results-SEPtoNC} without the condition on the $K_\lambda$ to be in the linear hull of $\Lambda_B(\rho_{AB})$. Clearly, the condition is met if the set $\Lambda_B(\rho_{AB})$ spans the entire operator space. If $\dim(\Ha_A) = \dim(\Ha_B)$, then for almost all separable states $\tau_{AB}$ we have that $\Lambda_B(\tau_{AB})$ spans the entire operator space, but note that this is not true for the maximally mixed state. We will use this insight to get rid of the pathological cases. The proof can be found in the Appendix~\ref{appendix:thm-results-SEPtoRobustNC}.

\begin{restatable}{theorem}{resultsSEPtoRobustNC} \label{thm:results-SEPtoRobustNC}
Let $\dim(\Ha_A) = \dim(\Ha_B)$ and let $\rho_{AB}$ be a separable quantum state. Then for almost all separable quantum states $\tau_{AB}$ there is a $\delta(\tau_{AB})>0$ depending on $\tau_{AB}$, such that for every $\varepsilon \in (0, \delta)$ there exists a P\&M noncontextual model for $\Lambda_A[(1-\varepsilon) \rho_{AB} + \varepsilon \tau_{AB}]$.
\end{restatable}

\paragraph{Mapping noncontextuality inequalities to entanglement witnesses}
Using the results of Theorem~\ref{thm:results-SEPtoRobustNC} one can obtain entanglement witnesses from noncontextuality inequalities. The only caveat is that the noncontextuality inequalities must be formulated in terms of unnormalized states. This is necessary to account for the fact that $\Tr_B[(\I_A \otimes E_B) \rho_{AB}]$ is not normalized for $E_B \geq 0$. This is because the set of $E_B$ such that $\Tr[(\I_A \otimes E_B) \rho_{AB}] = 1$ in general depends on $\rho_{AB}$.

We will demonstrate the method using the noncontextuality inequality presented in Ref.~\cite{MazurekPuseyKunjwalReschSpekkens-noncontextuality}. Let $K$ be the set of allowed preparations. By $\cone(K)$ we denote the set of all unnormalized allowed preparations, that is all operators of the form $\mu \tilde{\sigma}$, where $\mu \in \RR$, $\mu \geq 0$ and $\tilde{\sigma} \in K$. Let $\sigma_{t,b} \in \cone(K)$ for $t \in \{1,2,3\}$ and $b \in \{0,1\}$ be such that $\sigma_* = \frac{1}{2} (\sigma_{t,0} + \sigma_{t,1})$ is the same for all $t\in \{1,2,3\}$ and let $M_{t,b}$ be positive operators, $M_{t,b} \geq 0$, such that $\frac{1}{3} \sum_{t=1}^3 M_{t,b} = \frac{\I}{2}$ and $M_{t,0} + M_{t,1} = \I$. In other words, $M_{t,b}$ are three binary POVMs such that their uniform mixture corresponds to the random coin toss. Then the unnormalized version of the noncontextuality inequality from Ref.~\cite{MazurekPuseyKunjwalReschSpekkens-noncontextuality} is as follows: If there is a P\&M noncontextual model for $K$, then we have $\sum_{t=1}^3 \sum_{b=0}^1 \Tr(\sigma_{t,b} M_{t,b}) \leq 5 \Tr(\sigma_*)$. See the Appendix~\ref{appendix:prop-ineqToWitness} for the proof of this modified noncontextuality inequality. Using this inequality we obtain the following entanglement witness:
\begin{restatable}{proposition}{ineqToWitness} \label{prop:ineqToWitness}
Let $\rho_{AB}$ be a separable quantum state. Let $E_{t,b}$ and $M_{t,b}$ be positive operators, $E_{t,b} \geq 0$ and $M_{t,b} \geq 0$, such that $E_* = \frac{1}{2}(E_{t,0} + E_{t,1})$, $\frac{1}{3} \sum_{t=1}^3 M_{t,b} = \frac{\I_A}{2}$ and $\I_A = M_{t,0} + M_{t,1}$ for all $t\in \{1,2,3\}$ and $b \in \{0,1\}$. Then
\begin{equation}
\label{eq:ineqToWitness-ineqSEP}
\sum_{t=1}^3 \sum_{b=0}^1
\Tr[(M_{t,b} \otimes E_{t,b}) \rho_{AB}]
\leq
5 \Tr[(\I_A \otimes E_*) \rho_{AB}].
\end{equation}
Moreover, there is an entangled state $\rho_{AB}$ that violates Eq.~\eqref{eq:ineqToWitness-ineqSEP} for suitable choice of the operators $E_{t,b}$ and $M_{t,b}$.
\end{restatable}

The proof follows from Theorem~\ref{thm:results-SEPtoRobustNC}, since if $\rho_{AB}$ is separable, then Eq.~\eqref{eq:ineqToWitness-ineqSEP} is satisfied for all $\Lambda_A[(1-\varepsilon) \rho_{AB} + \varepsilon \tau_{AB}]$ for all $\varepsilon \in (0, \delta)$. This is then used to construct the corresponding entanglement witness. The full proof can be found in the Appendix~\ref{appendix:prop-ineqToWitness}.

Being more concrete, one can write down an explicit entanglement witness from observables leading to a violation of the original noncontextuality inequality \cite{BudroniCabelloGuhneKleinmann-contextuality}, see the Appendix~\ref{appendix:prop-ineqToWitness} for detailed construction. We obtain that for every separable state $\rho_{AB}$ we have $\Tr[(\sigma_x \otimes \sigma_x) \rho_{AB}] + \Tr[(\sigma_z \otimes \sigma_z) \rho_{AB}] \leq \frac{4}{3}$ which is a weakened version of the well-known witness $\Tr[(\sigma_x \otimes \sigma_x) \rho_{AB}] + \Tr[(\sigma_z \otimes \sigma_z) \rho_{AB}] \leq 1$ \cite{Toth-witness, BruknerVedral-witness, DowlingDohertyBarlett-witness}; still, this inequality is violated by the maximally entangled state.

\paragraph{Mapping entanglement witnesses to noncontextuality inequalities}
Noncontextuality inequalities are inequalities that are satisfied by any operational model that is P\&M noncontextual. Since they are theory independent, noncontextuality inequalities are experimentally feasible method of certifying non-classicality of experimental setups without having to trust the physical implementation. We will also sketch other potential uses for noncontextuality inequalities in quantum key distribution inspired by our results in later section. While some techniques for constructing noncontextuality inequalities are known \cite{SchmidSpekkensWolfe-noncontextualityIneq,KrishnaSpekkensWolfe-noncontextualityIneq}, our results enable us to derive additional noncontextuality inequalities. This is not so straightforward, as we cannot use just a single entanglement witness, but we must map a whole class of entanglement witnesses to get a class of noncontextuality inequalities. We will proceed with an example that showcases this. We will use a class of entanglement witnesses that comes from the Clauser-Horne-Shimony-Holt (CHSH) inequality \cite{ClauserHorneShimonyHolt-CHSH, ClauserHorneShimonyHolt-CHSHerratum}, but similar a approach works for other Bell inequalities as well \cite{PlavalaGuhne-NCineqFromBell}.

Let $A_i$ and $B_i$ for $i \in \{1,2\}$ be observables such that $-\I_A \leq A_i \leq \I_A$ and $-\I_B \leq B_i \leq \I_B$ and let $\rho_{AB}$ be a separable state. Then we have $\Tr\{[A_1 \otimes (B_1 + B_2) + A_2 \otimes (B_1 - B_2)] \rho_{AB}\} \leq 2$. In order to obtain a noncontextuality inequality proceed as follows. Let $B_{i+}$ and $ B_{i-}$ be the the positive and negative parts of $B_i$, respectively, that is, $B_{i+}$ and $ B_{i-}$ are operators such that $B_{i\pm} \geq 0$, $B_{i+} B_{i-} = 0$ and $B = B_{i+} - B_{i-}$. Denote $\sigma_{i\pm} = \Tr_B[(\I_A \otimes B_{i\pm}) \rho_{AB}]$. Then we have
\begin{widetext}
\begin{equation}
\label{eq:witnessToIneq-correspondence}
2 \geq \Tr\{[A_1 \otimes (B_1 + B_2) + A_2 \otimes (B_1 - B_2)] \rho_{AB}\}
= \Tr[(A_1 + A_2)(\sigma_{1+} - \sigma_{1-})]
+ \Tr[(A_1 - A_2)(\sigma_{2+} - \sigma_{2-})],
\end{equation}
\end{widetext}
which bears already some formal similarity to the noncontextuality inequality from above. From $-\I_B \leq B_i \leq \I_B$ it follows that the eigenvalues of $B_i$ are from the interval $[-1, 1]$ and so we also have $\abs{B_i} \leq \I_B$. Define $\sigma_{i 0} = \Tr_B\{[\I_A \otimes (\I_B - \abs{B}_i)] \rho_{AB}\}$ and $\sigma_* = \Tr_B(\rho_{AB})$, then we have $\sigma_{1+} + \sigma_{1-} + \sigma_{10} = \sigma_* = \sigma_{2+} + \sigma_{2-} + \sigma_{20}$, which is going to play a crucial role in the formulation of the noncontextuality inequality. We obtain:
\begin{restatable}{proposition}{witnessToIneq} \label{prop:witnessToIneq}
Let $K$ be a set of allowed preparations. Let $\sigma_* \in K$ and let $i \in \{1,2\}$, let $\sigma_{i+}, \sigma_{i-}, \sigma_{i0} \in \cone(K)$ be subnormalized preparations such that
\begin{equation}
\label{eq:witnessToIneq-constraint}
\sigma_{1+} + \sigma_{1-} + \sigma_{10}
= \sigma_*
= \sigma_{2+} + \sigma_{2-} +
\sigma_{20}.
\end{equation}
Let $A_i$ be observables such that $-\I \leq A_i \leq \I$ for all $i \in \{1,2\}$. If there is a P\&M contextual model for $K$, then
\begin{equation} \label{eq:witnessToIneq-ineqNC}
\Tr[(A_1 + A_2)(\sigma_{1+} - \sigma_{1-})] + \Tr[(A_1 - A_2)(\sigma_{2+} - \sigma_{2-})] \leq 2.
\end{equation}
\end{restatable}

The proof of Proposition~\ref{prop:witnessToIneq} is given in the Appendix~\ref{appendix:prop-witnessToIneq}. The proof is significantly different from the proof of Proposition~\ref{prop:ineqToWitness}: There, we showed that Eq.~\eqref{eq:ineqToWitness-ineqSEP} is an entanglement witness because it was derived from noncontextuality inequality. In the proof of Proposition~\ref{prop:witnessToIneq} we use Eq.~\eqref{eq:witnessToIneq-correspondence} only as an educated guess and we have to prove that Eq.~\eqref{eq:witnessToIneq-ineqNC} is a noncontextuality inequality by showing that it holds whenever a P\&M noncontextual model exists.

In order to construct an explicit violation of the noncontextuality inequality we can consider the equivalent of the standard quantum violation of the CHSH inequality: let $\Ha$ be a qubit Hilbert space, $\dH = 2$, let $\sigma_x, \sigma_y, \sigma_z$ be the Pauli matrices and let
\begin{align}
&A_1 = \frac{1}{\sqrt{2}} (\sigma_x + \sigma_z),
&&A_2 = \frac{1}{\sqrt{2}} (\sigma_x - \sigma_z), \nonumber \\
&\sigma_{1+} = \frac{1}{2} \ketbra{+},
&&\sigma_{1-} = \frac{1}{2} \ketbra{-}, \label{eq:ineqToWitness-example} \\
&\sigma_{2+} = \frac{1}{2} \ketbra{0},
&&\sigma_{2-} = \frac{1}{2} \ketbra{1}, \nonumber
\end{align}
where $\ket{+}, \ket{-}$ and $\ket{0}, \ket{1}$ are the eigenbasis of the Pauli operators $\sigma_x, \sigma_z$ respectively. We have $\sigma_{1+} + \sigma_{1-} = \sigma_{2+} + \sigma_{2-}$ and so the constraint \eqref{eq:witnessToIneq-constraint} is satisfied. We get $\Tr((A_1 + A_2)(\sigma_{1+} - \sigma_{1-})) + \Tr((A_1 - A_2)(\sigma_{2+} - \sigma_{2-})) = 2 \sqrt{2}$, hence the inequality \eqref{eq:witnessToIneq-ineqNC} is violated.

Let us note that it was shown in Ref.~\cite{SchmidSelbyWolfeKunjwalSpekkens-noncontextuality} that the stabilizer rebit theory, whose state space consists of the convex combinations of the states $\ketbra{0}$, $\ketbra{1}$, $\ketbra{+}$, $\ketbra{-}$, has a P\&M noncontextual model which may seem to contradict the violation of the noncontextuality inequality \eqref{eq:witnessToIneq-ineqNC}. However, there is no contradiction because the observables $A_1$ and $A_2$ are not included in the stabilizer rebit theory.

\paragraph{Applications in quantum key distribution}
Our results provide an insight into why entanglement is a precondition for secure quantum key distribution \cite{CurtyLewensteinLu-entanglementForQKD}. This is a known result, but we will show that it can be understood in terms of noncontextuality by employing our result. Consider an entanglement-based protocol for quantum key distribution, let $\rho_{AB}$ be the state shared between Alice and Bob and let $M_{a|x}$ be the measurements available to Alice and $N_{b|y}$ be the measurements available to Bob. In line with \cite{CurtyLewensteinLu-entanglementForQKD} we say that $M_{a|x} \otimes N_{b|y}$ do not witness the entanglement of $\rho_{AB}$ if for any entanglement witness of the form $W = \sum_{a,b,x,y} c_{abxy} M_{a|x} \otimes N_{b|y}$, that is for every operator $W$ of this form such that $\Tr(W \sigma_{AB}) \geq 0$ for all separable states $\sigma_{AB}$, we also have that $\Tr(W \rho_{AB}) \geq 0$. It was observed in \cite{CurtyLewensteinLu-entanglementForQKD} that if $M_{a|x} \otimes N_{b|y}$ do not witness the entanglement of $\rho_{AB}$, then there is a separable state $\xi_{AB}$ such that $\Tr[(M_{a|x} \otimes N_{b|y}) \rho_{AB}] = \Tr[(M_{a|x} \otimes N_{b|y}) \xi_{AB}]$.

We will interpret Bob's measurements on $\rho_{AB}$ as remote preparations, thus Bob will be sending to Alice states $\sigma_{a|x}$ proportional to $\Tr_B[(\I \otimes M_{a|x}) \rho_{AB}]$. If we assume that $M_{a|x} \otimes N_{b|y}$ do not witness the entanglement of $\rho_{AB}$, then we can replace $\rho_{AB}$ with the respective separable state $\xi_{AB}$ and it follows from Theorem~\ref{thm:results-SEPtoRobustNC} that there exists P\&M noncontextual model for the set $\Lambda_A[(1-\varepsilon) \xi_{AB} + \varepsilon \tau_{AB}]$. Then in any prepare and measure quantum key distribution protocol where Bob sends to Alice states from the set $\Lambda_A[(1-\varepsilon) \xi_{AB} + \varepsilon \tau_{AB}]$ for arbitrary small $\varepsilon > 0$, the eavesdropper can then map the state prepared by Bob to the hidden variables, broadcast the hidden variables, and map one copy to the respective state that is then send to Alice while using the other to obtain complete information about the secret key. It follows that quantum key distribution is not possible in this scenario, and thus it is also not possible with the state $\rho_{AB}$ and measurements $M_{a|x}$ and $N_{b|y}$ respectively.

The fact that a state is entangled does not guarantee that it is useful for quantum key distribution in entanglement based protocols and, analogically, the fact that $K$ is contextual does not guarantee that it is useful for quantum key distribution in prepare and measure protocols. But it is known that sufficient violation of a Bell inequality, such as the maximal violation of the CHSH inequality, is sufficient for quantum key distribution. Using the standard connection between entanglement based and prepare and measure protocols for quantum key distribution \cite{BennettBrassardMermin-QKD,ScaraniBechmannpasquinucciCerfDusekLutkenhausPeev-QKD,FerencziLutkenhaus-QKD,ColesMetodievLutkenhaus-numericalQKD}, it follows that sufficient violation of the respective noncontextuality inequality enables quantum key distribution. We leave further investigation of this topic for future research.

\paragraph{Conclusions}
Our main results, Theorems \ref{thm:results-NCtoSEP} and \ref{thm:results-SEPtoRobustNC}, prove that contextuality is a precondition for entanglement and that only entangled states allow for robust remote preparations of contextuality. We have used these results to map noncontextuality inequalities to a class of entanglement witnesses in Proposition~\ref{prop:ineqToWitness} and to design a class of noncontextuality inequalities using entanglement witnesses in Proposition~\ref{prop:witnessToIneq}. We have also discussed potential future applications of our results in quantum key distribution.

As a consequence of our results any experiment which verifies entanglement of a state (e.g., by observing quantum steering or violation of a Bell inequality) immediately verifies P\&M contextuality of the induced system. Moreover our results open the path to further transport of results between entanglement and contextuality, one can use our result to design noncontextual inequalities that can be used to verify parameters of an experimental setup \cite{PlavalaGuhne-NCineqFromBell}, and it is also possible to take a contextuality-enabled task and transform it into a remote entanglement-enabled task. Thus our results provide a blueprint for connecting the resource theory of entanglement and contextuality.

Note added: after submitting this manuscript for publication, a paper connecting contextuality and Bell non-locality in a similar manner was made public on arXiv and subsequently published \cite{WrightFarkas-contextuality}, where also the decidability of the membership problem of the set of quantum contextual behaviours is discussed.

\paragraph{Acknowledgments}
We thank Carlos de Gois, Benjamin Yadin, and Matthias Kleinmann for providing feedback on readability of the manuscript.

We acknowledge support from the Deutsche Forschungsgemeinschaft (DFG, German Research Foundation, project numbers 447948357 and 440958198), the Sino-German Center for Research Promotion (Project M-0294), the ERC (Consolidator Grant 683107/TempoQ), the German Ministry of Education and Research (Project QuKuK, BMBF Grant No. 16KIS1618K), and the Alexander von Humboldt Foundation.

\onecolumngrid
\appendix

\section{Proof of Theorem~\ref*{thm:results-NCtoSEP}} \label{appendix:thm-results-NCtoSEP} 
\resultsNCtoSEP*
\begin{proof}
Let $\sigma_A \in \Lambda_A(\rho_{AB})$ and let $\oM = \{M_a\}$ be a POVM on $\Ha_A$. According to our assumptions there exists P\&M noncontextual hidden variable model for $\Lambda_A(\rho_{AB})$, thus we have $\Tr( \sigma_A M_a) = \sum_\lambda \Tr(\sigma_A N_\lambda) \Tr(\omega_\lambda M_a)$ for some $\omega_\lambda \in \dens(\Ha_A)$ and $N_\lambda \in \bound(\Ha_A)$ such that $\Tr(\sigma_A' N_\lambda) \geq 0$ for all $\sigma_A' \in \Lambda_A(\rho_{AB})$. Since $\sigma_A \in \Lambda_A(\rho_{AB})$ then there is some $E_B \geq 0$ such that $\sigma_A = \Tr_B[(\I_A \otimes E_B) \rho_{AB}]$ and we have $\Tr(\sigma_A M_a) = \Tr[(M_a \otimes E_B) \rho_{AB}]$ and $\Tr(\sigma_A N_\lambda) = \Tr[(N_\lambda \otimes E_B) \rho_{AB}]$. Moreover we will use that
$\Tr[(N_\lambda \otimes E_B) \rho_{AB}] = \Tr \{ E_B \Tr_A [ (N_\lambda \otimes \I_B) \rho_{AB} ] \}$. Putting everything together we get
\begin{equation}
\begin{split}
\Tr[(M_a \otimes E_B) \rho_{AB}] &= \sum_\lambda \Tr[(N_\lambda \otimes E_B) \rho_{AB}] \Tr(\omega_\lambda M_a) \\
&= \sum_\lambda \Tr [ ( M_a \otimes N_\lambda \otimes E_B ) ( \omega_\lambda \otimes \rho_{AB} ) ] \\
&= \sum_\lambda \Tr [ ( M_a \otimes E_B ) ( \omega_\lambda \otimes \Tr_A [( N_\lambda \otimes \I_B) \rho_{AB}] ) ].
\end{split}
\end{equation}
Since this holds for every $M_a$ and $E_B$ we must have $\rho_{AB} = \sum_\lambda \omega_\lambda \otimes \Tr_A [( N_\lambda \otimes \I_B) \rho_{AB}]$. In order to prove that $\rho_{AB}$ is separable we only need to prove that $\Tr_A [( N_\lambda \otimes \I_B) \rho_{AB}]$ is positive semidefinite for all $\lambda$, but this is straightforward as for any positive semidefinite $F_B \in \bound(\Ha_B)$ we have
\begin{equation}
\Tr \{ F_B \Tr_A [( N_\lambda \otimes \I_B) \rho_{AB}] \} = \Tr \{ N_\lambda \Tr_B [( \I_A \otimes F_B) \rho_{AB}] \} = \mu \Tr(N_\lambda \sigma_A') \geq 0
\end{equation}
where $\sigma_A' \in \Lambda_A(\rho_{AB})$ and $\mu \geq 0$ are such that $\Tr_B [( \I_A \otimes F_B) \rho_{AB}] = \mu \sigma_A'$.
\end{proof}

\section{Proof of Theorem~\ref*{thm:results-SEPtoNC}} \label{appendix:thm-results-SEPtoNC}
Before proceeding to the proof of Theorem~2 we need to introduce two superoperators. Let $\rho_{AB} \in \dens(\Ha_A \otimes \Ha_B)$ be a bipartite state, then the superoperator $\Psi_{A \to B}^{\rho}: \bound(\Ha_A) \to \bound(\Ha_B)$ is defined as follows: let $X_A \in \bound(\Ha_A)$, then
\begin{equation}
\Psi_{A \to B}^{\rho}(X_A) = \Tr_A[(X_A \otimes \I_B) \rho_{AB}].
\end{equation}
The superoperator $\Psi_{B \to A}^{\rho}: \bound(\Ha_B) \to \bound(\Ha_A)$ is defined analogically. Moreover we can express the remotely preparable sets using the superoperator $\Psi_{B \to A}^{\rho}$ as $\Lambda_A(\rho_{AB}) = \{ \sigma_A \in \dens(\Ha_A): \: \sigma_A = \Psi_{B \to A}^{\rho}(E_B), E_B \geq 0 \}$.

Let $\{X_i^A\}_{i=1}^{\dim(\Ha_A)^2}$ and $\{Y_j^B\}_{j=1}^{\dim(\Ha_B)^2}$ be an orthonormal operator basis of $B(\Ha_A)$ and $B(\Ha_B)$ respectively, i.e., $\Tr(X_i^A X_j^A) = \Tr(Y_i^B Y_j^B) = \delta_{ij}$, then the basis can be chosen such that $\rho_{AB} = \sum_{i=1}^N p_i X_i^A \otimes Y_i^B$, where $N = \max[\dim(\Ha_A)^2, \dim(\Ha_B)^2]$ and $p_i \geq 0$, this is essentially Schmidt decomposition applied to $\rho_{AB}$ as a vector in $\bound(\Ha_A) \otimes \bound(\Ha_B)$. Let $I_+$ be the index set defined as $I_+ = \{ i: p_i > 0\}$. We then define the superoperator $\Pi_A^\rho: \bound(\Ha_A) \to \bound(\Ha_A)$ as follows:
\begin{equation} \label{eq:Pi-def}
\Pi_A^\rho (\sum_{i=1}^{\dim(\Ha_A)^2} \alpha_i X_i^A ) = \sum_{i \in I_+} \alpha_i X_i^A.
\end{equation}
$\Pi_B^\rho: \bound(\Ha_B) \to \bound(\Ha_B)$ is defined analogically.

The superoperators $\Pi_A^\rho$ and $\Pi_B^\rho$ will appear in our calculations for the following reason: the superoperator $\Psi_{B \to A}^{\rho}$ is not necessarily invertible, but we can invert it on its support. $\Pi_A^\rho$ and $\Pi_B^\rho$ are projections on the support of $\Psi_{B \to A}^{\rho}$ and so they will appear because we will use the pseudo-inverse of $\Psi_{B \to A}^{\rho}$.

\begin{lemma*}
Let $\rho_{AB}$ be a separable bipartite quantum state with the decomposition $\rho_{AB} = \sum_\lambda \omega_\lambda \otimes K_\lambda$, where $\omega_\lambda \geq 0$, $\Tr(\omega_\lambda) = 1$ and $K_\lambda \geq 0$. Assume that we have $\Pi_B^\rho(K_\lambda) \geq 0$ for all $\lambda$. Then there exists P\&M noncontextual model for $\Lambda_A(\rho_{AB})$.
\end{lemma*}
\begin{proof}
Let $\{X_i^A\}_{i=1}^{\dim(\Ha_A)^2}$ and $\{Y_j^B\}_{j=1}^{\dim(\Ha_B)^2}$ be orthonormal operator basis of $B(\Ha_A)$ and $B(\Ha_B)$ respectively such that $\rho_{AB} = \sum_{i=1}^N p_i X_i^A \otimes Y_i^B$, where $N = \max[\dim(\Ha_A)^2, \dim(\Ha_B)^2]$ and $p_i \geq 0$. Let $\sigma_A \in \Lambda_A(\rho_{AB})$ and let $E_B \in \bound(\Ha_B)$, $E_B \geq 0$, be the corresponding operator such that $\sigma_A = \Tr_B[(\I_A \otimes E_B) \rho_{AB}]$. Then $E_B = \sum_{i=1}^{\dim(\Ha_B)^2} \beta_i Y_i^B$ and we have
\begin{equation}
\sigma_A = \Tr_B[(\I_A \otimes E_B) \rho_{AB}] = \Psi_{B \to A}^{\rho}(E_B) = \sum_{i \in I_+} \beta_i p_i X_i^A.
\end{equation}
We will now define a superoperator $\Phi: \bound(\Ha_A) \to \bound(\Ha_B)$ that will act as the pseudo-inverse to the superoperator $\Psi_{B \to A}^{\rho}$ as follows: let $F_A \in \bound(\Ha_A)$, then $F_A = \sum_{i=1}^{\dim(\Ha_A)^2} \alpha_i X_i^A$ and we define $\Phi(F_A) = \sum_{i \in I_+} \frac{\alpha_i}{p_i} Y_i^B$. We clearly have
\begin{equation} \label{eq:results-SEPtoNC-pseudoInverse}
(\Phi \circ \Psi_{B \to A}^{\rho})(E_B) = \Pi_B^\rho(E_B),
\end{equation}
where $\circ$ denotes the concatenation of superoperators. Moreover it follows that $\Psi_{B \to A}^{\rho}(E_B) = (\Psi_{B \to A}^{\rho} \circ \Pi_B^\rho)(E_B)$. Let $\oM = \{M_a\}$ be a POVM, then we have
\begin{equation}
\Tr(\sigma_A M_a) = \Tr[(M_a \otimes E_B) \rho_{AB}]
= \Tr\{[M_a \otimes \Pi_B^\rho(E_B)] \rho_{AB}\}
= \sum_\lambda \Tr(\omega_\lambda M_a) \Tr[\Pi_B^\rho(E_B) K_\lambda].
\end{equation}
Using Eq.~\eqref{eq:results-SEPtoNC-pseudoInverse} we get
\begin{equation}
\Tr(\sigma_A M_a) = \sum_\lambda \Tr[\Psi_{B \to A}^{\rho}(E_B) \Phi^*(K_\lambda)] \Tr(\omega_\lambda M_a)
= \sum_\lambda \Tr[\sigma_A \Phi^*(K_\lambda)] \Tr(\omega_\lambda M_a)
\end{equation}
where $\Phi^*$ is the adjoint superoperator to $\Phi$. Denoting $\Phi^*(K_\lambda) = N_\lambda$ we get
\begin{equation}
\Tr(\sigma_A M_a) = \sum_\lambda \Tr(\sigma_A N_\lambda) \Tr(\omega_\lambda M_a)
\end{equation}
which is the desired result, we only need to check that $N_\lambda$ satisfies the positivity and normalization conditions.

Let $\sigma_A \in \Lambda_A(\rho_{AB})$ be given as $\sigma_A = \Psi_{B \to A}^{\rho}(E_B)$ for some $E_B \in \bound(\Ha_B)$, $E_B \geq 0$, then we have
\begin{equation}
\Tr(\sigma_A N_\lambda) = \Tr[\Phi(\sigma_A) K_\lambda] = \Tr[\Pi_B^\rho(E_B) K_\lambda] = \Tr[E_B \Pi_B^\rho(K_\lambda)] \geq 0
\end{equation}
since $\Pi_B^\rho(K_\lambda) \geq 0$ according to the assumptions of the theorem. To check normalization note that due to the adjoint version of Eq.~\eqref{eq:results-SEPtoNC-pseudoInverse} we have that $\Phi^*$ is the pseudo-inverse of the superoperator $\Psi_{A \to B}^{\rho}$. Using $\sum_\lambda K_\lambda = \Tr_A( \rho_{AB}) = \Tr_A[(\I_A \otimes \I_B) \rho_{AB}]$ we get that $\Phi^*(\sum_\lambda K_\lambda) = \Pi_A^\rho(\I_A)$ and so we have $\sum_\lambda N_\lambda = \Phi^*(\sum_\lambda K_\lambda) = \Pi_A^\rho(\I_A)$. It follows that $\sum_\lambda \Tr(\sigma_A N_\lambda) = \Tr[\sigma_A \Pi_A^\rho(\I)] = \Tr(\sigma_A) = 1$ because $\Pi_A^\rho(\sigma_A) = \sigma_A$.
\end{proof}

\resultsSEPtoNC*
\begin{proof}
The result follows from the previous lemma. Since for every $\sigma \in \Lambda_B(\rho_{AB})$ we have $\Pi^\rho_B(\sigma) = \sigma$, it follows that if $K_\lambda$ belongs to the linear hull of of $\Lambda_B(\rho_{AB})$ then we have $\Pi^\rho_B(K_\lambda) = K_\lambda \geq 0$.
\end{proof}

\section{Proof of Theorem~\ref*{thm:results-SEPtoRobustNC}} \label{appendix:thm-results-SEPtoRobustNC}
\resultsSEPtoRobustNC*
\begin{proof}
We will denote $n = \dim(\Ha_A) = \dim(\Ha_B)$. And let $\tau_{AB} \in \dens(\Ha_A \otimes \Ha_B)$ be a randomly selected separable state. Since $\dim(\Ha_A) = \dim(\Ha_B) = n$, we can represent the superoperators $\Psi_{B \to A}^\rho$ and $\Psi_{B \to A}^\tau$ by $n^2 \times n^2$ matrices $M(\Psi_{B \to A}^\rho)$ and $M(\Psi_{B \to A}^\tau)$. Moreover since $\tau_{AB}$ is randomly selected, we have $\det[M(\Psi_{B \to A}^\tau)] \neq 0$, since almost all matrices have non-zero determinants.

It follows that the matrix corresponding to the superoperator $\Psi_{B \to A}^{(1-\varepsilon) \rho + \varepsilon \tau}$ is $(1-\varepsilon) M(\Psi_{B \to A}^\rho) + \varepsilon M(\Psi_{B \to A}^\tau)$. In order to finish the proof we want to show that there is some $\delta \in (0,1)$ such that for all $\varepsilon \in (0,\delta)$ we have $\det[(1-\varepsilon) M(\Psi_{B \to A}^\rho) + \varepsilon M(\Psi_{B \to A}^\tau)] \neq 0$, then it follows that the superoperator $\Psi_{B \to A}^{(1-\varepsilon) \rho_{AB} + \varepsilon \tau_{AB}}$ is invertible.

We have that $\det[(1-\varepsilon) M(\Psi_{B \to A}^\rho) + \varepsilon M(\Psi_{B \to A}^\tau)]$ is either constant in $\varepsilon$ or a polynomial of finite order. Since the superoperator $\Psi_{B \to A}^\tau$ is invertible we must have $\det[M(\Psi_{B \to A}^\tau)] \neq 0$. Therefore if $\det[(1-\varepsilon) M(\Psi_{B \to A}^\rho) + \varepsilon M(\Psi_{B \to A}^\tau)]$ is constant, then $\det[(1-\varepsilon) M(\Psi_{B \to A}^\rho) + \varepsilon M(\Psi_{B \to A}^\tau)] \neq 0$ for all $\varepsilon \in [0,1]$ and the result follows. If $\det[(1-\varepsilon) M(\Psi_{B \to A}^\rho) + \varepsilon M(\Psi_{B \to A}^\tau)]$ is not constant, then let $\delta$ be the smallest root of $\det[(1-\varepsilon) M(\Psi_{B \to A}^\rho) + \varepsilon M(\Psi_{B \to A}^\tau)]$ from the interval $(0,1)$. It follows that for all $\varepsilon \in (0, \delta)$ we have $\det[(1-\varepsilon) M(\Psi_{B \to A}^\rho) + \varepsilon M(\Psi_{B \to A}^\tau)] \neq 0$ and thus the corresponding map is invertible.
\end{proof}

One can construct suitable $\tau_{AB}$ explicitly as follows: Let $\frac{\I_A}{\sqrt{n}}, \tilde{X}_i^A \in \bound(\Ha_A)$ and $\frac{\I_B}{\sqrt{n}}, \tilde{Y}_i^B \in \bound(\Ha_B)$ be orthonormal basis, i.e., $\Tr(\tilde{X}_i^A) = \Tr(\tilde{Y}_i^B) = 0$ and $\Tr(\tilde{X}_i^A \tilde{X}_j^A) = \Tr(\tilde{Y}_i^A \tilde{Y}_j^B) = \delta_{ij}$. Note that these are different basis than the Schmidt basis used in previous constructions. We then define $\tau_{AB} \in \bound(\Ha_A \otimes \Ha_B)$ as $\tau_{AB} = \frac{\I_A \otimes \I_B}{n} + \mu \sum_i \tilde{X}_i^A \otimes \tilde{Y}_i^B$, where $\mu > 0$. It follows that we can always choose $\mu$ so small that $\tau_{AB} \geq 0$, moreover such that $\tau_{AB}$ is a separable state. Let $E_B \in \bound(\Ha_B)$, then $E_B = \beta_1 \frac{\I_B}{\sqrt{n}} + \sum_i \beta_i \tilde{Y}_i^B$ and we have $\Tr_B((\I_A \otimes E_B) \tau_{AB}) = \beta_1 \frac{\I_A}{\sqrt{n}} + \mu \sum_i \beta_i \tilde{X}_i^B$, from which it follows that the superoperator $\Psi_{B \to A}^{\tau}$ is invertible. It then follows that $\Psi_{B \to A}^{(1-\varepsilon) \rho + \varepsilon \tau}$ is also invertible for $\varepsilon \in (0, \delta)$ for suitable choice of $\delta$.

\section{Proof of Proposition~\ref*{prop:ineqToWitness} and explicit construction of an entanglement witness} \label{appendix:prop-ineqToWitness}
The following result is a restatement of a known noncontextuality inequality \cite{MazurekPuseyKunjwalReschSpekkens-noncontextuality}, our only modification is that we allow for unnormalized states.
\begin{proposition*}
Let $K \subset \dens(\Ha)$ be the set of allowed preparations and let $\cone(K)$ denote the set of all unnormalized allowed preparations, that is all operators of the form $\mu \tilde{\sigma}$, where $\mu \in \RR$, $\mu \geq 0$ and $\tilde{\sigma} \in K$. Let $\sigma_{t,b} \in \cone(K)$ for $t \in \{1,2,3\}$ and $b \in \{0,1\}$ be such that $\sigma_* = \frac{1}{2} (\sigma_{t,0} + \sigma_{t,1}) = \frac{1}{2} (\sigma_{t',0} + \sigma_{t',1})$ for all $t,t' \in \{1,2,3\}$ and let $M_{t,b} \in \bound(\Ha)$ be positive operators, $M_{t,b} \geq 0$, such that $\frac{1}{3} \sum_{t=1}^3 M_{t,b} = \frac{\I}{2}$ and $M_{t,0} + M_{t,1} = \I$ for all $t$. If there is preparation and measurement noncontextual hidden variable model for $K$, then we have
\begin{equation} \label{eq:ineqToWitness-ineqNC}
\sum_{t=1}^3 \sum_{b=0}^1 \Tr(\sigma_{t,b} M_{t,b}) \leq 5 \Tr(\sigma_*).
\end{equation}
\end{proposition*}
\begin{proof}
We get
\begin{equation}
\begin{split}
\sum_{t=1}^3 \sum_{b=0}^1 \Tr(\sigma_{t,b} M_{t,b}) & = \sum_{t=1}^3 \sum_{b=0}^1 \sum_\lambda \Tr(\sigma_{t,b} N_\lambda) \Tr( \omega_\lambda M_{t,b})
\leq \sum_{t=1}^3 \sum_\lambda \max_b[\Tr( \omega_\lambda M_{t,b})] \sum_{b=0}^1 \Tr(\sigma_{t,b} N_\lambda) \\
 & = 2 \sum_{t=1}^3 \sum_\lambda \max_b[\Tr( \omega_\lambda M_{t,b})] \Tr(\sigma_* N_\lambda)
\leq 6 \max_\lambda \{ \frac{1}{3} \sum_{t=1}^3 \max_b[\Tr( \omega_\lambda M_{t,b})] \} \sum_\lambda \Tr(\sigma_* N_\lambda) \\
 & \leq 5 \Tr(\sigma_*),
\end{split}    
\end{equation}
where we have used that $\frac{1}{3} \sum_{t=1}^3 \max_b[\Tr( \omega_\lambda M_{t,b})] \leq \frac{5}{6}$, which is a standard step in proving the origin noncontextuality inequality, see \cite{MazurekPuseyKunjwalReschSpekkens-noncontextuality, BudroniCabelloGuhneKleinmann-contextuality}. To see that $\sum_\lambda \Tr(\sigma_* N_\lambda) = \Tr(\sigma_*)$ simply note that $\sigma_* \in \cone(K)$ and so $\sigma_* = \Tr(\sigma_*) \tilde{\sigma}_*$ where $\tilde{\sigma}_* \in K$. Then $\sum_\lambda \Tr(\sigma_* N_\lambda) = \Tr(\sigma_*) \sum_\lambda \Tr(\tilde{\sigma}_* N_\lambda) = \Tr(\sigma_*)$ as a result of the normalization of $N_\lambda$.
\end{proof}

\ineqToWitness*
\begin{proof}
The proof is straightforward: assume that $\rho_{AB}$ is separable and let $\tau_{AB}$ be a separable state such that there exists preparation and measurement noncontextual hidden variable model for $\Lambda_A((1-\varepsilon) \rho_{AB} + \varepsilon \tau_{AB})$ for all $\varepsilon \in (0, \delta)$ for suitable $\delta \in (0, 1)$. Denote $\sigma_{t,b,\varepsilon} = \Tr_B((\I_A \otimes E_{t,b})((1-\varepsilon) \rho_{AB} + \varepsilon \tau_{AB}))$ and $\sigma_{*,\varepsilon} = \Tr_B((\I_A \otimes E_{*})((1-\varepsilon) \rho_{AB} + \varepsilon \tau_{AB}))$, then Eq.~\eqref{eq:ineqToWitness-ineqNC} becomes $\sum_{t=1}^3 \sum_{b=0}^1 \Tr(\sigma_{t,b,\varepsilon} M_{t,b}) \leq 5 \Tr(\sigma_{*,\varepsilon})$. We thus get
\begin{equation}
\sum_{t=1}^3 \sum_{b=0}^1 \Tr\{[E_{t,b} \otimes M_{t,b}][(1-\varepsilon) \rho_{AB} + \varepsilon \tau_{AB}]\} \leq 5 \Tr\{[E_* \otimes \I_B][(1-\varepsilon) \rho_{AB} + \varepsilon \tau_{AB}]\}
\end{equation}
for all $\varepsilon \in (0, \delta)$. Taking the limit $\varepsilon \to 0^+$ yields Eq.~\eqref{eq:ineqToWitness-ineqSEP}.

To show that there is an entangled state that violates Eq.~\eqref{eq:ineqToWitness-ineqSEP} simply assume that $\sigma_{t,b}$ and $M_{t,b}$ are states and POVMs that violate Eq.~\eqref{eq:ineqToWitness-ineqSEP}, it is known that such states and POVMs exist \cite{MazurekPuseyKunjwalReschSpekkens-noncontextuality, BudroniCabelloGuhneKleinmann-contextuality}. Let $\ket{\phi^+} = \frac{1}{\sqrt{\dim(\Ha)}} \sum_{i=1}^{\dim(\Ha)} \ket{ii}$ be the maximally entangled state and take $\rho_{AB} = \ketbra{\phi^+}$. Take $E_{t,b} = \dim(\Ha) \sigma_{t,b}^\intercal$, where $A^\intercal$ is the transposition of $A$ with respect to the basis $\ket{i}$. Then we have $\Tr_B[(\I_A \otimes E_{t,b}) \rho_{AB}] = \sigma_{t,b}$ and so Eq.~\eqref{eq:ineqToWitness-ineqSEP} must be violated because the corresponding noncontextuality inequality is violated.
\end{proof}

To construct explicit example of an entanglement witness, we consider qubit systems and the Pauli matrices $\sigma_x, \sigma_y, \sigma_z$. Then we define the effects
\begin{align}
 & E_{1,0} = \I + \sigma_z,
 & & E_{1,1} = \I - \sigma_z, \\
 & E_{2,0} = \I + \frac{\sqrt{3}}{2} \sigma_x - \frac{1}{2} \sigma_z,
 & & E_{2,1} = \I - \frac{\sqrt{3}}{2} \sigma_x + \frac{1}{2} \sigma_z, \\
 & E_{3,0} = \I - \frac{\sqrt{3}}{2} \sigma_x - \frac{1}{2} \sigma_z,
 & & E_{3,1} = \I + \frac{\sqrt{3}}{2} \sigma_x + \frac{1}{2} \sigma_z,
\end{align}
and $M_{t,b} = \frac{1}{2} E_{t,b}$. From this we obtain the entanglement witness $W = 2 \I_A \otimes \I_B - \frac{3}{2} ( \sigma_x \otimes \sigma_x + \sigma_z \otimes \sigma_z )$, so for every separable state $\rho_{AB}$ we have
\begin{equation}
\Tr[(\sigma_x \otimes \sigma_x) \rho_{AB}] + \Tr[(\sigma_z \otimes \sigma_z) \rho_{AB}] \leq \frac{4}{3}.
\end{equation}

\section{Proof of Proposition~\ref*{prop:witnessToIneq}} \label{appendix:prop-witnessToIneq}
\witnessToIneq*
\begin{proof}
In order to shorten the notation let us denote
\begin{equation}
\Xi(A,\sigma) = \Tr[(A_1 + A_2)(\sigma_{1+} - \sigma_{1-})] + \Tr[(A_1 - A_2)(\sigma_{2+} - \sigma_{2-})].
\end{equation}
The proof is straightforward: assume that there is preparation and measurement noncontextual hidden variable model for $K \subset \dens(\Ha)$. Then we have $\Tr( \sigma_{i \pm} A_j) = \sum_\lambda \Tr(\sigma_{i \pm} N_\lambda) \Tr(\omega_\lambda A_j)$ for all $i,j \in \{1,2\}$, this follows for example from $A_i = 2 M_i - \I$ where $0 \leq M_i \leq \I$ is an appropriate operator. We thus get
\begin{equation}
\begin{split}
\Xi(A,\sigma) & = \sum_\lambda \Tr[N_\lambda (\sigma_{1+} - \sigma_{1-})] \Tr[(A_1 + A_2) \omega_\lambda] + \sum_\lambda \Tr[N_\lambda (\sigma_{2+} - \sigma_{2-})] \Tr[(A_1 - A_2) \omega_\lambda] \\
& = \sum_\lambda \Tr[N_\lambda (\sigma_{1+} - \sigma_{1-} + \sigma_{2+} - \sigma_{2-})] \Tr(A_1 \omega_\lambda) + \sum_\lambda \Tr[N_\lambda (\sigma_{1+} - \sigma_{1-} - \sigma_{2+} + \sigma_{2-})] \Tr(A_2 \omega_\lambda) \\
& \leq \sum_\lambda \abs{\Tr[N_\lambda (\sigma_{1+} - \sigma_{1-} + \sigma_{2+} - \sigma_{2-})]} + \sum_\lambda \abs{\Tr[N_\lambda (\sigma_{1+} - \sigma_{1-} - \sigma_{2+} + \sigma_{2-})]},
\end{split}    
\end{equation}
where we have used that $\Tr(A_i \omega_\lambda) \in [-1,1]$ for all $i \in \{1,2\}$ and all $\lambda$. Let us fix $\lambda$ and let us inspect the term $\abs{\Tr[N_\lambda (\sigma_{1+} - \sigma_{1-} + \sigma_{2+} - \sigma_{2-})]} + \abs{\Tr[N_\lambda (\sigma_{1+} - \sigma_{1-} - \sigma_{2+} + \sigma_{2-})]} = \xi_\lambda$. There are four options on how the signs can be assigned: if $\Tr[N_\lambda (\sigma_{1+} - \sigma_{1-} + \sigma_{2+} - \sigma_{2-})] \geq 0$ and $\Tr[N_\lambda (\sigma_{1+} - \sigma_{1-} - \sigma_{2+} + \sigma_{2-})] \geq 0$ we get
\begin{equation}
\xi_\lambda = \Tr[N_\lambda (\sigma_{1+} - \sigma_{1-} + \sigma_{2+} - \sigma_{2-})] + \Tr[N_\lambda (\sigma_{1+} - \sigma_{1-} - \sigma_{2+} + \sigma_{2-})] = 2 \Tr[ N_\lambda (\sigma_{1+} - \sigma_{1-})] \leq 2 \Tr(N_\lambda \sigma_*),
\end{equation}
where the last inequality follows from Eq.~\eqref{eq:witnessToIneq-constraint} and from the positivity condition. If $\Tr[N_\lambda (\sigma_{1+} - \sigma_{1-} + \sigma_{2+} - \sigma_{2-})] \geq 0$ and $\Tr[N_\lambda (\sigma_{1+} - \sigma_{1-} - \sigma_{2+} + \sigma_{2-})] < 0$ we get
\begin{equation}
\xi_\lambda = \Tr[N_\lambda (\sigma_{1+} - \sigma_{1-} + \sigma_{2+} - \sigma_{2-})] - \Tr[N_\lambda (\sigma_{1+} - \sigma_{1-} - \sigma_{2+} + \sigma_{2-})] = 2 \Tr[ N_\lambda (\sigma_{2+} - \sigma_{2-})] \leq 2 \Tr(N_\lambda \sigma_*),
\end{equation}
if $\Tr[N_\lambda (\sigma_{1+} - \sigma_{1-} + \sigma_{2+} - \sigma_{2-})] < 0$ and $\Tr[N_\lambda (\sigma_{1+} - \sigma_{1-} - \sigma_{2+} + \sigma_{2-})] \geq 0$ we get
\begin{equation}
\xi_\lambda = - \Tr[N_\lambda (\sigma_{1+} - \sigma_{1-} + \sigma_{2+} - \sigma_{2-})] + \Tr[N_\lambda (\sigma_{1+} - \sigma_{1-} - \sigma_{2+} + \sigma_{2-})] = 2 \Tr[ N_\lambda (\sigma_{2-} - \sigma_{2+})] \leq 2 \Tr(N_\lambda \sigma_*),
\end{equation}
if $\Tr[N_\lambda (\sigma_{1+} - \sigma_{1-} + \sigma_{2+} - \sigma_{2-})] < 0$ and $\Tr[N_\lambda (\sigma_{1+} - \sigma_{1-} - \sigma_{2+} + \sigma_{2-})] < 0$ we get
\begin{equation}
\xi_\lambda = - \Tr[N_\lambda (\sigma_{1+} - \sigma_{1-} + \sigma_{2+} - \sigma_{2-})] - \Tr[N_\lambda (\sigma_{1+} - \sigma_{1-} - \sigma_{2+} + \sigma_{2-})] = 2 \Tr[ N_\lambda (\sigma_{1-} - \sigma_{1+})] \leq 2 \Tr(N_\lambda \sigma_*).
\end{equation}
We thus have $\Xi(A,\sigma) \leq 2 \sum_\lambda \Tr(N_\lambda \sigma_*) = 2$ where we have used that $\sigma_* \in K$ and the normalization condition.
\end{proof}

\bibliography{citations}

\end{document}